\documentclass[12pt]{article}
\usepackage{amssymb,amsmath,amsfonts}
\usepackage{amsthm} 
\usepackage{esvect} 
\usepackage{bm} 
\usepackage{graphicx} 
\usepackage{comment}
\usepackage{lmodern}
\usepackage{todonotes}
\usepackage{url}
\usepackage{pdfsync}
\usepackage{adjustbox}
\usepackage{multirow}
\usepackage{bbm}
\usepackage{algorithm}
\usepackage[noend]{algpseudocode}
\usepackage{amsopn}

\newcommand{\E}{{\mathbb E}}

\newcommand{\Q}{{\mathbb Q}}
\newcommand{\C}{{\mathbb C}}
\newcommand{\R}{{\mathbb R}}
\newcommand{\N}{{\mathbb N}}
\newcommand{\Z}{{\mathbb Z}}

\newcommand{\Fcal}{{\mathcal F}}
\newcommand{\Ucal}{{\mathcal U}}
\newcommand{\Scal}{{\mathcal S}}

\newtheorem{proposition}{Proposition}[section]
\newtheorem{lemma}[proposition]{Lemma}
\newtheorem{theorem}[proposition]{Theorem}
\newtheorem{definition}[proposition]{Definition}
\newtheorem{corollary}[proposition]{Corollary}
\newtheorem{remark}[proposition]{Remark}

\begin{document}

\title{Unspanned Stochastic Volatility in the Multi-Factor CIR Model\footnote{We thank Scott Joslin, Anders Trolle, and two anonymous referees for their comments. The research leading to these results has received funding from the European Research Council under the European Union's Seventh Framework Programme (FP/2007-2013) / ERC Grant Agreement n. 307465-POLYTE.}}
\author{Damir Filipovi\'c\footnote{EPFL and Swiss Finance Institute, 1015 Lausanne, Switzerland. {\it Email: }damir.filipovic@epfl.ch} \and Martin Larsson\footnote{ETH Zurich, Department of Mathematics, R\"amistrasse 101, CH-8092, Zurich, Switzerland. {\it Email: }martin.larsson@math.ethz.ch} \and Francesco Statti \footnote{EPFL, 1015 Lausanne, Switzerland. {\it Email: }francesco.statti@epfl.ch}}
\date{April 13, 2018\\forthcoming in {\it Mathematical Finance}
}

\maketitle

\begin{abstract}
Empirical evidence suggests that fixed income markets exhibit unspanned stochastic volatility (USV), that is, that one cannot fully hedge volatility risk solely using a portfolio of bonds. While \cite{collin2002bonds} showed that no two-factor Cox-Ingersoll-Ross (CIR) model can exhibit USV, it has been unknown to date whether CIR models with more than two factors can exhibit USV or not. We formally review USV and relate it to bond market incompleteness. We provide necessary and sufficient conditions for a multi-factor CIR model to exhibit USV. We then construct a class of three-factor CIR models that exhibit USV. This answers in the affirmative the above previously open question. We also show that multi-factor CIR models with diagonal drift matrix cannot exhibit USV.
\end{abstract}




\section{Introduction}

Empirical evidence suggests that fixed income markets exhibit unspanned stochastic volatility (USV), i.e., that one cannot fully hedge volatility risk solely using a portfolio of bonds, see \cite{collin2002bonds,joslin2017can,trolle2009unspanned,filipovic2017linear}. One of the basic models for the term structure of interest rates is the multi-factor Cox--Ingersoll--Ross (CIR) model. While \cite{collin2002bonds} showed that no two-factor CIR model can exhibit USV, it has been unknown to date whether CIR models with more than two factors can exhibit USV or not.

In this paper, we give necessary and sufficient conditions for a multi-factor CIR model to exhibit USV. These conditions reveal that multi-factor CIR models do not exhibit USV in general. We show that the number of USV factors in a $d$-factor CIR model is limited by $d-2$. For $d=2$ this confirms the finding of \cite{collin2002bonds}. We then construct a class of three-factor CIR models that exhibit USV. This answers in the affirmative whether CIR models with more than two factors can exhibit USV or not.

The first systematic analysis of USV in affine term structure models was done in \cite{collin2002bonds}. They also give empirical evidence for state variables that drive innovations in interest rate derivatives, but do not affect innovations in the term structure of bond prices. They then identify a class of affine term structure models that can exhibit USV. In a similar vein, \cite{joslin2017can} characterizes a large class of affine term structure models with USV. He shows that the USV condition implies cutting edge restrictions on the model parameters. This is in contrast to the linear-rational term structure models introduced in \cite{filipovic2017linear} that can generically exhibit USV. A specific affine term structure model for commodities that exhibits USV was introduced in \cite{trolle2009unspanned}. Our paper complements this literature, as the USV models in \cite{collin2002bonds,joslin2017can,trolle2009unspanned,filipovic2017linear} do not contain the multi-factor CIR models.

The structure of the paper is as follows. In Section~\ref{secUSV} we formally define USV and relate it to bond market incompleteness in a multi-factor short rate model. In Section~\ref{secCIR} we provide necessary and sufficient conditions for USV in a multi-factor CIR model. In Section~\ref{sec3CIR} we construct a three-factor CIR model that exhibits USV. In Section~\ref{secdiag} we show that multi-factor CIR models with diagonal drift matrix cannot exhibit USV.

\section{USV in Multi-Factor Short Rate Models}\label{secUSV}

Throughout we fix a filtered probability space $(\Omega,\Fcal,\Fcal_t,\Q)$ where $\Q$ denotes the risk-neutral pricing measure. We consider a multi-factor short rate model in the following sense, see, e.g., \cite{duf_kan_96}. Let $E\subset\R^d$ be a convex state space for some $d\in \N$. Let $X$ be an $E$-valued Markov diffusion factor process of the form
\begin{equation}\label{srm1}
dX_t=b(X_t)dt+ \sigma(X_t) dW_t,
\end{equation}
for some functions $b:E\to \R^d$ and $\sigma:E\to \R^{d\times d}$, and where $W$ is a $d$-dimensional Brownian motion.
We assume throughout that $\sigma(X_t)$ is invertible $dt\otimes d\Q$-a.e., and that the support of $X_t$ is all of $E$ for every $t>0$.
The short rate is given by
\begin{equation}\label{srm2}
r_t=\rho(X_t)
\end{equation}
for some function $\rho:E\to\R$. Due to the Markov property of $X$, the price at time $t\le T$ of a zero-coupon bond maturing at time $T$ is given by
\begin{equation}\label{Fdef}
P(t,T)= \E[e^{-\int_t^T r_s ds}|\mathcal{F}_t] =F(T-t,X_t),
\end{equation}
for some function $F$ on $\R_+\times E$, which we assume to be $C^{1,2}$.

We now define the concept of term structure factors. We call $\xi\in\R^d\setminus\{0\}$ an unspanned direction if the term structure of bond prices $P(t,T)$, $T\ge t$, is unaffected by perturbations of $X_t$ along $\xi$. The linear span of all unspanned directions is called the term structure kernel and is denoted $\Ucal$. It is given by
\begin{equation}\label{Udef}
\Ucal = \bigcap_{\tau\ge0,\, x\in E} \ker \nabla_x F(\tau,x)^\top,
\end{equation}
where $\cdot^\top$ denotes the transpose, see also \cite{filipovic2017linear}. Let $m=d-\dim\Ucal\ge 0$, and fix a linear map $S:\R^d\to\R^m$ such that $\ker S=\Ucal$. In view of Lemma~\ref{lemfrepr} there exists a $C^{1,2}$-function $\tilde F$ on $\R_+\times S(E)$ such that $F(\tau,x)=\tilde F(\tau, S x)$ for all $\tau\ge 0$, $x\in E$. Defining
\begin{equation}\label{Zdef}
  Z_t = SX_t  ,
\end{equation}
it follows that the zero-coupon bond prices can be rewritten as
\begin{equation}\label{P(t,T)=G(T-t,Zt)}
P(t,T)= \tilde F(T-t, Z_t),
\end{equation}
so that, at any fixed time $t$, the term structure $P(t,T)$, $T\ge t$, is a function of $Z_t$ only. Note that also $r_t=\tilde\rho(Z_t)$ is a function of $Z_t$ only, where $\tilde\rho(z)=-\partial_\tau \tilde F(\tau,z)|_{\tau=0}$.\footnote{This uses that $\partial_T P(t,T)|_{T=t}=r_t$, which holds under mild assumptions, for instance continuity of $r$ and uniform integrability of $\{\tau^{-1}(\exp(-\int_t^{t+\tau} r_sds)-1)\mid \tau\in(0,\varepsilon)\}$ for all $t\ge0$ and some $\varepsilon>0$ that may depend on $t$. In particular, this holds in the CIR model.} This motivates the following terminology.
\begin{definition}
We refer to $Z_t$ as term structure factors and, accordingly, to $U_t=L X_t$ as unspanned factors, for any linear map $L:\R^d\to\R^{d-m}$ such that $L^\top\R^{d-m}=\Ucal$.
\end{definition}

We next show that the existence of unspanned directions, $\dim\Ucal >0$, can give rise to bond market incompleteness in the sense that not all European claims on the term structure can be replicated by solely trading in bonds and the money-market account. In view of \eqref{P(t,T)=G(T-t,Zt)}, any such claim has a payoff of the form $\Phi(Z_T)$ at some $T$. Due to the Markov property of $X$, the price at time $t\le T$ is given by
\[ \Pi_t = \E[e^{-\int_t^T r_s ds} \Phi(Z_T)|\mathcal{F}_t] = G(t,X_t)\]
for some function $G$ on $[0,T]\times E$. If $G$ is $C^{1,2}$, we say that $\Phi(Z_T)$ is a regular claim. In this case It\^o's formula yields
\begin{equation} \label{Pi_t}
d\Pi_t = r_t \Pi_t dt + \nabla_x G(t,X_t)^\top \sigma(X_t)dW_t.
\end{equation}
On the other hand, it follows from \eqref{Zdef} and \eqref{P(t,T)=G(T-t,Zt)} that the value process $V$ of any self-financing trading strategy in bonds and the money-market account is of the form
\begin{equation} \label{V_t}
dV_t = r_t V_t dt + \theta_t^\top S\sigma(X_t)dW_t,
\end{equation}
where $\theta$ is an $\R^m$-valued progressively measurable process.

As $\sigma(X_t)$ is invertible $dt\otimes d\Q$-a.e.\ and the support of $X_t$ is all of $E$ for every $t>0$, we infer from \eqref{Pi_t} and \eqref{V_t} that a regular claim $\Phi(Z_T)$ can be replicated if and only if $\nabla_x G(t,x)\in S^\top\R^m$ for all $t\in[0,T]$, $x\in E$. Combining this with Lemma~\ref{lemfrepr} we obtain the following result.

\begin{lemma}\label{lemRepReg}
A regular claim $\Phi(Z_T)$ can be replicated if and only if
\[ \E[e^{-\int_t^T r_s ds} \Phi(Z_T)|\mathcal{F}_t] = \tilde G(t,Z_t)\]
for some $C^{1,2}$-function $\tilde G$ on $[0,T]\times S(E)$. This holds in particular if $Z$ itself is a Markov process.
\end{lemma}

The concept of unspanned stochastic volatility (USV) is now made precise by the following definition.

\begin{definition} \label{defUSV}
The $d$-factor short rate model~\eqref{srm1}--\eqref{srm2} exhibits USV if the bond market is incomplete in the sense that there exists a regular claim $\Phi(Z_T)$ that cannot be replicated.
\end{definition}

\begin{remark}
Different choices of $S$ lead to term structure factors $Z_t$ that are related by linear bijections of $\R^m$. Thus the definition of USV does not depend on the specific choice of $S$.
\end{remark}

As empirical evidence suggests that fixed income markets exhibit USV, it is also a desirable feature of term structure models. However, it turns out that multi-factor short rate models do not generically exhibit USV. Indeed, if $m=d$, then $Z$ is a linear bijective transformation of $X$. In this case there are no unspanned directions, let alone USV. Even if $m<d$, then there are unspanned directions but not necessarily USV. To see this, let $r_t=Z_t$, where $X=(Z,U)$ is a bivariate CIR process (see following section) with independent components. In particular $Z$ is a Markov process, so that in view of Lemma~\ref{lemRepReg}, the model does not exhibit USV. This agrees with the intuition that $U$ is an irrelevant factor that has no influence whatsoever on the term structure.

\section{Multi-Factor CIR Models and USV}\label{secCIR}

An important example of a multi-factor short rate model \eqref{srm1}--\eqref{srm2} is the Cox--Ingersoll--Ross (CIR) model, see e.g.\ \cite{MR2553163} for details. The $d$-factor CIR model consists of the $E=\mathbb{R}^d_+$-valued square-root diffusion factor process $X$ with dynamics of the form
\begin{equation}\label{cir1}
dX_t=(b+\beta X_t)dt+ {\rm diag}(\sigma_1 \sqrt{X_{1t}}, \dots, \sigma_d \sqrt{X_{dt}})dW_t,
\end{equation}
for some $b \in \mathbb{R}^d_+,$ $\beta \in \mathbb{R}^{d \times d}$ with nonnegative off-diagonals, $\beta_{ij} \geq 0$ for $i \neq j$, and $\sigma_i> 0$. Here $X_{1t},\ldots,X_{dt}$ denote the components of $X_t$. The short rate is given by
\begin{equation}\label{cir2}
r_t=\rho^\top X_t
\end{equation}
for some parameter $\rho \in \mathbb{R}^d_+\setminus\{0\}$. The price at time $t$ of a zero-coupon bond maturing at time $T$ is given by \eqref{Fdef} with the exponential affine function of the form
\[ F(\tau,x)=e^{-A(\tau)-B(\tau)^\top x}.\]
The $\R$- and $\R^d$-valued functions $A(\tau)$ and $B(\tau)$ solve the Riccati equations
\begin{equation}\label{Riccati}
\begin{aligned}
\partial_\tau A(\tau)&=b^\top B(\tau),&A(0)=0, \\
\partial_\tau B(\tau)&=H(B(\tau)), \quad&B(0)=0,
\end{aligned}
\end{equation}
where we define the map
\begin{equation*}
H: \mathbb{R}^d \to \mathbb{R}^d, \quad H(v)=-\frac{1}{2} \sigma^2 \circ v \circ v +\beta^\top v + \rho,
\end{equation*}
where $\circ$ denotes component-wise multiplication (Hadamard product) and $\sigma^2=(\sigma_1^2,\dots,\sigma_d^2)^\top$. The term structure kernel \eqref{Udef} becomes
\[ \Ucal = \bigcap_{\tau\ge0} \ker B(\tau)^\top.\]
Let $m=d-\dim\Ucal$ and $S:\R^d\to\R^m$ be a linear map with $\ker S=\Ucal$ as above.\footnote{The function $\tilde F$ in \eqref{P(t,T)=G(T-t,Zt)} can be chosen as $ \tilde F(\tau,z) = \exp(-A(\tau)-B(\tau)^\top Q  z)$ where $Q=S^\top(SS^\top)^{-1}:\R^m\to\R^d$ so that $ B(\tau)^\top QS=B(\tau)^\top$.}
An equivalent condition for USV in the CIR model is given by the following result.

\begin{theorem}\label{thm2.1}
The $d$-factor CIR model \eqref{cir1}--\eqref{cir2} exhibits USV if and only if
\begin{equation}\label{conditionLemmaNEW}
H(S^\top \mathbb{R}^m) \not\subseteq S^\top \mathbb{R}^m.
\end{equation}
\end{theorem}

\begin{proof}
For any $u \in \R^d$, let $\phi(\tau, u)$ and $\psi (\tau, u)$ be the solution of the following system of Riccati differential equations:
\begin{align*}
\partial_\tau\phi(\tau,u)&=b^\top \psi(\tau,u), & \phi(0,u)=0,\\
\partial_\tau\psi(\tau,u)&=H(\psi(\tau,u)), \quad& \psi(0,u)=u,
\end{align*}
so that $A(\tau)=-\phi(\tau,0)$ and $B(\tau)=-\psi(\tau,0)$. Then, for any $x \in \mathbb{R}^d$, $v \in \R^m$, and $t\ge0$ such that the left-hand side is finite, we have
\begin{align}\label{affine}
\mathbb{E}_x\left[ e^{-\int_0^t r_s ds} e^{v^\top Z_t}\right]=\mathbb{E}_{x}\left[ e^{-\int_0^t r_s ds} e^{(S^\top v)^\top X_t}\right]
=e^{\phi(t,S^\top v)+\psi(t,S^\top v)^\top x},
\end{align}
where in the last equality we apply the affine property of $X$; see e.g.~\cite{MR1994043}. This shows in particular that $\Phi(Z_t)=e^{v^\top Z_t}$ is a regular claim.

For any $v\in\R^m$, there is an open interval $I\subset \R$ containing zero such that \eqref{affine} holds for all $t\in I$. If the CIR model \eqref{cir1}--\eqref{cir2} does not exhibit USV then the last quantity in \eqref{affine} depends on $x$ only through the value of $z=Sx$. Perturbing $x$ by elements of $\ker S$, we thus obtain
\begin{equation*}
 \psi(t,S^\top v) \in (\ker S)^\perp = S^\top \mathbb{R}^m, \qquad t\in I.
\end{equation*}
This implies $H(S^\top v)=\partial_\tau\psi(\tau,S^\top v)|_{\tau=0} \in S^\top \mathbb{R}^m$, hence \eqref{conditionLemmaNEW} does not hold.

Conversely, if \eqref{conditionLemmaNEW} does not hold, then $\psi(t,S^\top v)$ lies in $S^\top \R^m$ for all $t\ge0$ and $v\in\R^m$ such that this quantity exists, and is therefore equal to $S^\top \tilde\psi(t,v)$ for some $\tilde\psi(t,v)\in\R^m$. Thus the left-hand side of \eqref{affine} is a function of $z=Sx$ only. This shows that $Z$ is a Markov process, and hence the CIR model \eqref{cir1}--\eqref{cir2} does not exhibit USV.
\end{proof}

Theorem \ref{thm2.1} yields an important corollary, which shows that a CIR model needs at least two term structure factors in order to exhibit USV.

\begin{corollary}\label{cor1}
Whether the $d$-factor CIR model \eqref{cir1}--\eqref{cir2} exhibits USV or not depends on the model parameters $\sigma^2$, $\beta$, and $\rho$ only. There can be at most $d-2$ USV factors, so that necessarily the number of term structure factors satisfies $m \geq 2$.
\end{corollary}

\begin{proof}
The first statement follows directly from Theorem~\ref{thm2.1} and the fact that $H$ and $B$ only depend on $\sigma^2$, $\beta$, and $\rho$. For the second statement, we argue by contradiction and suppose that the $d$-factor CIR model exhibits USV with $m=1$. As $\rho=\partial_\tau B(\tau)|_{\tau=0}$ this implies
\begin{equation*}
S^\top \mathbb{R} =\text{span}(\rho),
\end{equation*}
and hence $\{ B(\tau)\mid\tau\ge 0\}\supset\{ s\rho\mid s\in I\}$ for some open interval $I\subset \R$ containing zero. Let $\xi\perp S^\top \mathbb{R}$. Then
\[ \xi^\top H(B(\tau))= \xi^\top \partial_\tau B(\tau) = 0\quad\text{for all $\tau\ge 0$},\]
and hence $\xi^\top H(s\rho)=0$ for all $s\in I$. As $H(w)$ is an analytic function of $w\in\R^d$ we conclude that $\xi^\top H(s\rho)=0$ for all $s\in\R$ and hence \eqref{conditionLemmaNEW} does not hold, which shows that USV fails.
\end{proof}

A simple consequence is stated in the following corollary, which confirms the finding of \cite{collin2002bonds}.
\begin{corollary}\label{corollary}
There exists no two-factor CIR model that exhibits USV.
\end{corollary}

An example of an alternative two-factor Markov model of the term structure that exhibits USV is given in \cite[Section II]{filipovic2017linear}.


\section{A Three-Factor CIR Model With USV}\label{sec3CIR}

We construct a three-factor CIR model that exhibits USV. Corollary~\ref{corollary} indicates that the dimension $d=3$ is the first nontrivial case that can be considered. Following Corollary \ref{cor1}, we aim at constructing a model with $m=2$ term structure factors. Here is our main result.

\begin{theorem}\label{thmEX}
The three-factor CIR model with $\sigma_i= \sqrt{2}$, $i=1,2,3$,
\[ \beta=
\left(
\begin{array}{ccc}
\beta_{11}&0&\beta_{13}\\
0 &\beta_{22} &\beta_{23}\\
0 & 0 &\beta_{33}
\end{array}
\right),\]
for parameters
\begin{equation}\label{betas}
\begin{aligned}
&\beta_{22} < \beta_{11} < 0,\quad \beta_{23}>0,\\
&\beta_{13}=\frac{8\rho_2}{\beta_{11}-\beta_{22}}+\beta_{23}-2\beta_{22},\\
&\beta_{33}=\beta_{11}+\beta_{22}-\frac{1}{2}(\beta_{13}+\beta_{23}),
\end{aligned}
\end{equation}
and
\begin{equation}\label{rho}
\rho= ( \rho_1, \rho_2, \rho_1 +\rho_2)^\top
\end{equation}
for parameters $\rho_2>0$ and
\begin{equation}\label{rho1}
 \rho_1=\frac{1}{8}(\beta_{11}-\beta_{22}) (\beta_{13}-\beta_{23}-2\beta_{11}),
\end{equation}
exhibits USV. Linear maps $S:\R^3\to\R^2$ and $L:\R^3\to\R$ with $\ker S=\Ucal$ and $L^\top\R=\Ucal$ are given by
\[  S=\begin{pmatrix}
1&0&1\\
0 &1 &1
\end{pmatrix},\quad L=\begin{pmatrix} 1 & 1 & -1\end{pmatrix}.
\]
The corresponding term structure and unspanned factors are $Z_t=SX_t=(X_{1t}+X_{3t}, X_{2t}+X_{3t})^\top$ and $U_t =LX_t = X_{1t}+X_{2t}-X_{3t}$. Note that $S(\R^3_+)=\R^2_+$, so that $Z$ is $\R^2_+$-valued, and $U\le Z_1+Z_2$.
\end{theorem}

Note that after component-wise scaling of the factors we can always normalize to $\sigma_i= \sqrt{2}$, $i=1,2,3$, without loss of generality. Moreover, \eqref{betas} and \eqref{rho} imply that $\beta_{13}>0,$ $ \beta_{33}<0$, and $\rho_1 >0$. Therefore, the corresponding CIR model is well defined and mean-reverting as the diagonal elements (eigenvalues) of $\beta$ are negative. While Theorem~\ref{thmEX} gives a parametric class of three-factor CIR models that exhibit USV, with four free parameters ($\beta_{11}$, $\beta_{22}$, $\beta_{23}$, $\rho_2$), the parameter constraints \eqref{betas}--\eqref{rho1} are knife-edge. This is in contrast to the linear-rational term structure models introduced in \cite{filipovic2017linear} that generically can exhibit USV.\footnote{The drift constraints in the linear-rational square-root (LRSQ) model in \cite[Section II]{filipovic2017linear} are straightforward such that the transformed process $Z_t$ has an autonomous drift.}

\begin{proof}[Proof of Theorem~\ref{thmEX}]

We have to show that $\Ucal=\ker S$ and that \eqref{conditionLemmaNEW} holds.

The condition $\Ucal=\ker S$ reads
\begin{equation} \label{eqB1+B2=B3}
B_3(\tau) = B_1(\tau) + B_2(\tau), \quad \tau\ge0,
\end{equation}
which in view of the relation $\rho=\partial_\tau B(\tau)|_{\tau=0}$ is consistent with \eqref{rho}. The assumed structure of $\beta$ enables us to rewrite \eqref{Riccati} as
\begin{align}\label{1equation_NEW}
\partial_\tau B_1(\tau)&=-B_1(\tau)^2+\beta_{11} B_1(\tau)+\rho_1,\\ \label{2equation_NEW}
\partial_\tau B_2(\tau)&=-B_2(\tau)^2+\beta_{22} B_2(\tau)+\rho_2,\\ \label{3equation_NEW}
\partial_\tau B_3(\tau)&=-B_3(\tau)^2+\beta_{13} B_1(\tau)+\beta_{23} B_2(\tau)+\beta_{33} B_3(\tau)+\rho_1+\rho_2,
\end{align}
Hence \eqref{eqB1+B2=B3} holds if and only if
\begin{align*}
\partial_\tau B_1(\tau)+\partial_\tau B_2(\tau) &= -(B_1(\tau)+B_2(\tau))^2+\beta_{13} B_1(\tau)\\
&\quad +\beta_{23} B_2(\tau) +\beta_{33} (B_1(\tau)+B_2(\tau))+\rho_1+\rho_2, \quad\tau\ge0.
\end{align*}
In view of \eqref{1equation_NEW} and \eqref{2equation_NEW}, this is equivalent to
\begin{equation}\label{conditionBi_NEW}
c_1B_1(\tau)+c_2B_2(\tau)-2B_1(\tau)B_2(\tau) = 0, \quad\tau\ge0,
\end{equation}
where $c_i=\beta_{i3} + \beta_{33} - \beta_{ii}$, $i=1,2$.

To prove that~\eqref{conditionBi_NEW} holds, we use that the solutions to \eqref{1equation_NEW} and \eqref{2equation_NEW} are given by
\[  B_i(\tau)=\frac{ 2\rho_i(e^{\theta_i\tau}-1)}{(\theta_i-\beta_{ii}) (e^{\theta_i\tau}-1) + 2\theta_i } ,\quad \theta_i=\sqrt{\beta_{ii}^2+4\rho_i},\quad i=1,2,\]
see \cite[Lemma 10.12]{MR2553163}. The form \eqref{betas} and \eqref{rho1} of $\beta$ and $\rho$ implies that $\theta=\theta_1=\theta_2$, and in order to simplify notation we write $B_i(\tau)= N_i(\tau)/D_i(\tau)$ with
\begin{align*}
N_i(\tau)&=2\rho_i(e^{\theta \tau}-1),\\
D_i(\tau)&=(\theta-\beta_{ii}) (e^{\theta\tau}-1) + 2\theta .
\end{align*}
With this notation, \eqref{conditionBi_NEW} can equivalently be written
\[
c_1N_1(\tau)D_2(\tau)+c_2N_2(\tau)D_1(\tau)-2N_1(\tau)N_2(\tau) = 0, \quad  \tau \geq 0,
\]
which upon inserting the expressions for $N_i(\tau)$ and $D_i(\tau)$ becomes
\begin{equation}\label{conditionBi3_NEW}
- \gamma_0+\gamma_1 e^{\theta \tau}+(\gamma_0-\gamma_1)e^{2\theta \tau} = 0, \quad \tau \geq 0,
\end{equation}
where
\begin{align*}
\gamma_0 &= 2c_1\rho_1(\theta+\beta_{22}) + 2c_2\rho_2(\theta+\beta_{11}) + 8\rho_1\rho_2, \\
\gamma_1 &= 4c_1\rho_1 \beta_{22} + 4c_2\rho_2\beta_{11} + 16\rho_1\rho_2.
\end{align*}
A further calculation shows that $\gamma_0=\gamma_1=0$ holds if
\begin{align*}
\beta_{13}+\beta_{23}&=2(\beta_{11}+\beta_{22}-\beta_{33}),\\
(\beta_{13}-\beta_{23})(\beta_{11}-\beta_{22})&=(\beta_{11} - \beta_{22})^2+4(\rho_1+\rho_2).
\end{align*}
This system is indeed satisfied by the model parameters $\beta$ and $\rho$ in \eqref{betas} and \eqref{rho1}. We conclude that \eqref{conditionBi3_NEW}, hence \eqref{conditionBi_NEW}, is satisfied, and hence $\Ucal=\ker S$.

It remains to verify that \eqref{conditionLemmaNEW} holds. Note that $S^\top\R^2=\ker L $. On the other hand, we have
\begin{equation*}
L H(S^\top v)=2v_1v_2+\ell(v),
\end{equation*}
for some first order polynomial $\ell(v)$ in $v$. The right hand side is certainly nonzero for some $v\in\R^2$, which shows \eqref{conditionLemmaNEW}.
\end{proof}

\begin{remark}
To see how the unspanned factor $U_t$ affects the bond return volatility, we calculate the quadratic variation of the log return, using \eqref{eqB1+B2=B3} and $U=Z_1+Z_2-3X_3$,
\begin{align*}
  \frac{1}{2}\frac{d\langle \log P(\cdot,T)\rangle_t}{dt} &= B(T-t)^\top \sigma(X_t)\sigma(X_t)^\top B(T-t)\\
  &= \sum_{i=1}^2 B_i(T-t)^2 X_{it}   + B_1(T-t)B_2(T-t) X_{3t}\\
  &= \sum_{i=1}^2 (B_i(T-t)^2-\frac{1}{3}) Z_{it}  - \frac{1}{3} B_1(T-t)B_2(T-t) U_{t}.
\end{align*}
Because $B_1 B_2 >0$, this reveals that there is USV, in line with Theorem~\ref{thmEX}.
\end{remark}

\section{CIR Models With Diagonal $\beta$}\label{secdiag}

In Theorem~\ref{thmEX} we assumed that $\beta_{13}, \beta_{23}>0$, so that $\beta$ was not diagonal. We now show that indeed there exists no CIR model with diagonal $\beta$ and USV.

Consider a $d$-factor CIR model \eqref{cir1}--\eqref{cir2} with diagonal $\beta$, which henceforth we parametrize as $\beta={\rm diag} (\beta_{1},\dots,\beta_{d})$. After component-wise scaling of $X$ we can assume that $\sigma_i= \sqrt{2}$, $i=1,\dots,d$, without loss of generality. We can also assume without loss of generality that $\rho_i>0$ for any $i=1,\dots,d$, because otherwise we could omit $X_i$ and the $d$-factor CIR model would in fact be a $(d-1)$-factor model.

The Riccati equations~\eqref{Riccati} fully decouple and the solutions $B_i$ are explicitly given by
\[  B_i(\tau)=\frac{ 2\rho_i(e^{\theta_i\tau}-1)}{(\theta_i-\beta_i) (e^{\theta_i\tau}-1) + 2\theta_i } ,\quad \theta_i=\sqrt{\beta_i^2+4\rho_i},\]
see \cite[Lemma 10.12]{MR2553163}. Note that $B_i(\tau)$ uniquely extends to an analytic function of $\tau\in\C$ with poles at $\tau\in \Scal_i$, where
\[ \Scal_i=\left\{ z\in\C\mid {\rm Re}\, z = \frac{1}{\theta_i}\log( \frac{\theta_i+\beta_i}{\theta_i-\beta_i}),\quad {\rm Im} \, z=\frac{1}{\theta_i}(2n+1)\pi,\quad n\in\Z\right\}. \]
There is a one-to-one correspondence between the sets of poles $\Scal_i$ and the parameters $(\theta_i,\beta_i)$ in the sense that $\Scal_i\cap\Scal_j\neq\emptyset$ if and only if $\Scal_i=\Scal_j$ if and only if $(\theta_i,\beta_i)=(\theta_j,\beta_j)$.
From this we draw two conclusions and our main result.

First, the functions $B_i$ are in one-to-one relation to the parameters $(\theta_i,\beta_i)$. That is, $B_i=B_j$ if and only if $(\theta_i,\beta_i)=(\theta_j,\beta_j)$, or equivalently, $(\rho_i,\beta_i)=(\rho_j,\beta_j)$. Now let $m\le d$ be number of elements of the set $\{B_1,\dots,B_d\}$. After reordering the indices, we can assume that $\{B_1,\dots,B_m\}=\{B_1,\dots,B_d\}$, so that $B_i\neq B_j$ for all $1\le i<j\le m$. Here is the second conclusion.
\begin{lemma}\label{lemLIND}
The functions $B_1,\dots,B_m$ are linearly independent.
\end{lemma}

\begin{proof}
Let $\zeta\in\R^m$ be such that $f(\tau)=\sum_{i=1}^m \zeta_i B_i(\tau)=0$ for all $\tau\in [0,\infty)$. By analytic continuation, $f(\tau)=0$ for all $\tau\in\C\setminus\cup_{i=1}^m\Scal_i$. On the other hand, $f(\tau)$ has a pole at $\tau\in\Scal_i$ if and only if $\zeta_i\neq 0$. Hence $\zeta=0$.
\end{proof}

Combining the above, we arrive at our main result.
\begin{theorem}
A CIR model with diagonal $\beta$ cannot exhibit USV.
\end{theorem}

\begin{proof}
Let $\{1,\dots,d\}=I_1\cup\cdots\cup I_m$ be the partition such that $B_i=B_k$ for all $i \in I_k$, $k=1,\dots,m$. We claim that
\[ \Ucal=\{ \xi\mid \text{$\textstyle\sum_{i\in I_k} \xi_i=0$ for all $k=1,\dots,m$}\}.\]
Indeed, $\xi\in\Ucal$ if and only if $\sum_{k=1}^m (\sum_{i\in I_k} \xi_i)  B_k  =0$, so that Lemma~\ref{lemLIND} yields the claim.

Hence a linear map $S:\R^d\to\R^m$ with $\ker S=\Ucal$ is given by $S_{ki}=1$ if $i\in I_k$ and $0$ otherwise. The corresponding term structure factors $Z_t=S X_t$ are given by $Z_{kt}=\sum_{i\in I_k}X_{it}$ and form a $m$-dimensional Markov process. Indeed, this follows from the independence of $X_1,\dots,X_d$ and because $\sigma_i=\sigma_j=\sqrt{2}$ and $\beta_i=\beta_j$ for all $i,j\in I_k$, see \cite[Corollary~10.4]{MR1994043}. By Lemma~\ref{lemRepReg}, the model therefore does not exhibit USV.
\end{proof}

\begin{appendix}

\section{Auxiliary Lemma}

\begin{lemma}\label{lemfrepr}
Let $f$ be a $C^1$-function on $E$ and $S:\R^d\to\R^m$ be a linear map with full rank, for some $0\le m\le d$. The following are equivalent:
\begin{enumerate}
  \item\label{lemfrepr1} $\nabla  f(x) \in S^\top\R^m$ for all $x\in E$;
  \item\label{lemfrepr2} $\ker S\subseteq\ker\nabla  f(x)^\top $ for all $x\in E$;
  \item\label{lemfrepr3} there exists a $C^1$-function $\tilde f$ on $S(E)$ such that $f(x)=\tilde f(Sx)$ for all $x\in E$.
\end{enumerate}
In either case, for any $z_0=Sx_0\in S(E)$ we have $\tilde f(z)=f(x_0+Q(z-z_0))$ for all $z\in \R^m$ such that $x_0+Q(z-z_0)\in E$, where $Q=S^\top(SS^\top)^{-1}$.
\end{lemma}

\begin{proof}
{\ref{lemfrepr1}}$\Leftrightarrow${\ref{lemfrepr2}}: trival.

{\ref{lemfrepr3}}$\Rightarrow${\ref{lemfrepr1}}: follows from the identity $\nabla  f(x)=S^\top \nabla \tilde f(Sx)$.

{\ref{lemfrepr2}}$\Rightarrow${\ref{lemfrepr3}}: we first claim that $f(x)=f(y)$ for all $x,y\in E$ such that $Sx=Sy$. Indeed, by convexity of $E$, we have that $x(\lambda)=\lambda x+(1-\lambda)y\in E$ for all $\lambda\in [0,1]$ and therefore
\[ \frac{d}{d\lambda} f(x(\lambda))= (x-y)^\top \nabla f(x(\lambda)) = 0\]
because $x-y\in\ker S $, which proves the claim. Hence, for any $z\in S(E)$, we can define $\tilde f(z)=f(x)$ for any $x\in E$ with $Sx=z$. The last statement of the lemma follows because $S(x_0+Q(z-z_0))=z$, which also shows that $\tilde f$ is $C^1$ on $S(E)$.
\end{proof}

\end{appendix}


\bibliographystyle{plain}
\bibliography{references}

\begin{thebibliography}{1}

\bibitem{collin2002bonds}
Pierre Collin-Dufresne and Robert~S Goldstein.
\newblock Do bonds span the fixed income markets? theory and evidence for
  unspanned stochastic volatility.
\newblock {\em The Journal of Finance}, 57(4):1685--1730, 2002.

\bibitem{MR1994043}
D.~Duffie, D.~Filipovi{\'c}, and W.~Schachermayer.
\newblock Affine processes and applications in finance.
\newblock {\em Ann. Appl. Probab.}, 13(3):984--1053, 2003.

\bibitem{duf_kan_96}
Darrell Duffie and Rui Kan.
\newblock A yield-factor model of interest rates.
\newblock {\em Mathematical Finance}, 6(4):379--406, 1996.

\bibitem{MR2553163}
Damir Filipovi{\'c}.
\newblock {\em Term-structure models}.
\newblock Springer Finance. Springer-Verlag, Berlin, 2009.
\newblock A graduate course.

\bibitem{filipovic2017linear}
Damir Filipovi{\'c}, Martin Larsson, and Anders~B Trolle.
\newblock Linear-rational term structure models.
\newblock {\em The Journal of Finance}, 72(2):655--704, 2017.

\bibitem{joslin2017can}
Scott Joslin.
\newblock Can unspanned stochastic volatility models explain the cross section
  of bond volatilities?
\newblock {\em Management Science}, forthcoming.

\bibitem{trolle2009unspanned}
Anders~B Trolle and Eduardo~S Schwartz.
\newblock Unspanned stochastic volatility and the pricing of commodity
  derivatives.
\newblock {\em Review of Financial Studies}, 22(11):4423--4461, 2009.

\end{thebibliography}

\end{document}